\documentclass[%
 reprint,
nofootinbib,
 amsmath,amssymb,
 aps, showkeys,
]{revtex4-2}

\usepackage{graphicx}
\usepackage[english]{babel}
\usepackage{amsthm}
\usepackage{txfonts}
\usepackage[colorlinks=true, allcolors=blue]{hyperref}
\usepackage{enumitem}
\usepackage{comment}
\usepackage{mathtools}
\usepackage{multirow}
\usepackage{cancel}
\usepackage[table]{xcolor}

\usepackage{tikz}
\usetikzlibrary{shapes.geometric, arrows, arrows.meta, positioning}

\theoremstyle{definition}
\newtheorem{theorem}{Theorem}
\newtheorem{lemma}{Lemma}
\newtheorem{definition}{Definition}

\def\ket#1{|#1\rangle}

\definecolor{myblue}{cmyk}  {1, 0, 0, 0}
\definecolor{mygreen3}{cmyk}{1, 0, 1, 0}
\definecolor{mygreen1}{cmyk}{.5,0, 1, 0}
\definecolor{myred}{cmyk}   {0, 1, 1, 0}
\definecolor{myyellow}{cmyk}{0, 0, 1, 0}

\newcommand*\cbb{\cellcolor{myblue!40}}   
\newcommand*\cgb{\cellcolor{mygreen3!40}} 
\newcommand*\cgr{\cellcolor{mygreen1!40}} 
\newcommand*\crr{\cellcolor{myred!70}}    

\newcommand\myH{\stackrel{\mathclap{\tiny\mbox{\text{H}}}}{\to}}
\newcommand\myC{\stackrel{\mathclap{\tiny\mbox{\text{CNOT}}}}{\to}}

\hypersetup{
pdftitle={On the quantum separability of qubit registers},
bookmarksopen=true,
}

\begin{document}

\preprint{APS/123-QED}

\title{On the quantum separability of qubit registers}

\author{Szymon Łukaszyk}
\email{szymon@patent.pl}
\affiliation{Łukaszyk Patent Attorneys, Głowackiego 8, 40-052 Katowice, Poland}


\begin{abstract}
We show that the bipartite separability of a pure qubit state hinges critically on the combinatorial structure of its computational‑basis support. Using Boolean cube geometry, we introduce a taxonomy that distinguishes support‑guaranteed separability from cases in which entanglement depends on probability amplitudes. We provide closed‑form support counts, identify forbidden configurations that enforce multipartite entanglement, and show how these results can enable fast entanglement diagnostics in quantum circuits. The framework offers immediate utility in classical simulation, entanglement‑aware circuit design, and quantum error‑correcting code analysis. This establishes support geometry as a practical and scalable tool for understanding entanglement in quantum information processing.
\end{abstract}

\keywords{
quantum entanglement,
quantum separability,
Boolean‑cube geometry,
quantum information processing,
combinatorial quantum information
}

\maketitle

\section{Introduction}\label{sec:Introduction}

Quantum entanglement is typically introduced through \textit{physical localization} of the quantum states in terms of quantum systems \textit{being in} or \textit{having} those states.
However, this phenomenon, which is, indeed, fundamentally nonlocal \cite{bell_einstein_1964, aspect_experimental_1982}, is not even required to demonstrate nonlocality \cite{wang_violation_2025}. Fundamental quantum mechanics has been constructed directly outside classical physics and even outside the general classical thinking \cite{mugur-schachter_crucial_2008}, and the mathematical concepts of quantum state entanglement and separability can be examined separately from the system's physical attributes.
This bypasses a cybernetic problem \cite{gershenson_self-organizing_2025} of defining a system and its boundaries, which is the root of various quantum \textit{paradoxes}, of which Schrödinger's cat is perhaps the most prominent example, as it requires a box in which a cat is \textit{localized}.

In this study, we examine the bipartite separability conditions of pure qubit quantum registers as rays in Hilbert spaces devoid of any \textit{spatial boundaries}.
Although a quantum state does not need to be described by qubits, any $m$-level quantum state (pure or mixed) for $m>2$ can be represented as a state on $\lceil \log_2(m) \rceil$ qubits via a substitution.
\begin{figure}[htbp]
\centering
\includegraphics[width=\columnwidth]{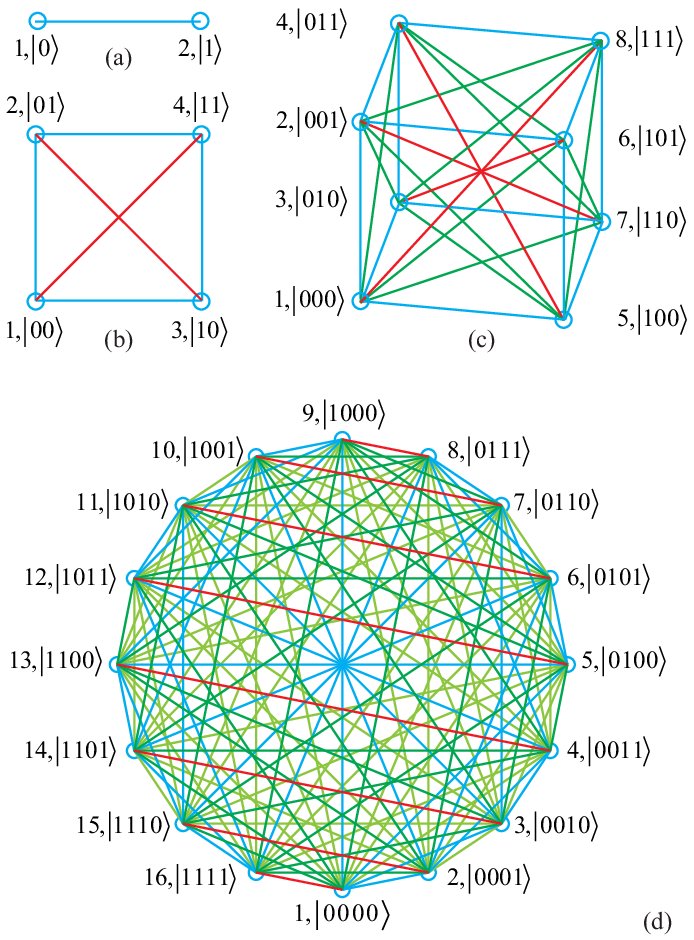}
\caption{\label{Fig:support}
Maximum supports of one (a) to four (d) qubit registers in the computational basis.
Vertices and blue edges represent any-partition-separable states, green edges --- partition-dependent separable states, and red edges --- states inseparable.
}
\end{figure}

We consider the states separable only across certain bipartitions, as well as those separable only for certain values of their probability amplitudes, and we provide their distributions with respect to their support sizes.
The results can be applied in quantum computational applications.

\section{Methods}\label{sec:Methods}

$\left\{0, 1\right\}^n$ Boolean space, $n \in \mathbb{N}$, can be considered as a complete graph constructed upon $n$-cube, where a distinct index $j=1,2,\cdots,2^n$ and a distinct address $a(j) = \{0,1\}^n$ can be assigned to each vertex \cite{SLcubes}.

Consider a pure quantum register containing $n$ qubits in the computational basis
\begin{equation}\label{eq:state}
\ket{A} = \sum_{j=1}^{2^n} \alpha_j \ket{a(j)}, \quad \sum_{j=1}^{2^n} |\alpha_j|^2 = 1,
\end{equation}
where
$\alpha_j \in \mathbb{C}$ are the probability amplitudes (PA, we sometimes write them as $\alpha_{a(j)}$),
$\ket{a(j)}$ are the basis components (kets in Dirac notation), and
$a(j)$ are their addresses.
Though the standard computational basis, with kets corresponding to the $j^{\text{th}}$ vertices of $n$-cube \cite{duff_black_2013}, applies naturally to quantum circuits, quantum error correction codes, etc., it can also be considered a physical aspect of nature \cite{lukaszyk_black_2025}. 

The \textit{support} of the state \eqref{eq:state} $\text{supp}(\ket{A})$ is the set of basis kets $\ket{a(j)}$ for which $\alpha_j \ne 0$, and the \textit{support size} ($k=|\text{supp}(\ket{A})|$) is the cardinality of this set.
The maximum support sizes for one to four qubits are shown in Fig.~\ref{Fig:support}.
There are $2^{2^n}-1$ distinct supports for $n$ qubits starting from $2^n$ one-ket states (single vertices of an $n$-cube) and ending on a state containing all $2^n$ kets (all vertices of the $n$-cube).

A pure quantum state $\ket{A}$ is called \textit{separable} if and only if the state~\eqref{eq:state} can be written as a tensor product
\begin{equation}\label{eq:separable_state}
\ket{A} = \ket{B} \otimes \ket{C},
\end{equation}
of at least two quantum states, which can be represented in the computational basis as similar qubit registers 
\begin{equation}
\begin{split}
\ket{B} &= \beta_1 \ket{0_1 0_2 \dots 0_{l-1} 0_l}       + \beta_2 \ket{0_1 0_2 \dots 0_{l-1} 1_l} + \dots \\
        &+ \beta_{2^l-1} \ket{1_1 1_2 \dots 1_{l-1} 0_l} + \beta_{2^l} \ket{1_1 1_2 \dots 1_{l-1} 1_l},\\
\ket{C} &= \gamma_1\ket{0_1 0_2 \dots 0_{m-1} 0_m}       + \gamma_2\ket{0_1 0_2 \dots 0_{m-1} 1_m} + \dots \\
        &+ \gamma_{2^m-1}\ket{1_1 1_2 \dots 1_{m-1} 0_m} + \gamma_{2^m}\ket{1_1 1_2 \dots 1_{m-1} 1_m},\\
\end{split}    
\end{equation}
consisting respectively of $l>0$ and $m>0$ qubits, where $n=l+m$ and $\sum_{j=1}^{2^l} |\beta_j|^2 = \sum_{j=1}^{2^m} |\gamma_j|^2 = 1$.

Otherwise, $\ket{A}$ is an \textit{entangled} state, and the degree of its entanglement can be measured using various methods. 
One of them employs the entanglement (von Neumann) entropy (here in bits)
\begin{equation} \label{eq:entanglement_entropy}
\mathcal{S}_A = - \sum_j \lambda_j \log_2 (\lambda_j), \quad 0 \le \mathcal{S}_A \le \min(l,m),
\end{equation}
where $\lambda_j \in \mathbb{R}_{\ge 0}$, $\sum_j \lambda_j=1$ are the eigenvalues of the reduced density matrix
\begin{equation}\label{eq:reduced_density_matrix}
\rho_{B} = \text{Tr}_C(\rho_{A}) \quad \text{or} \quad \rho_{C} = \text{Tr}_B(\rho_{A}).
\end{equation}
obtained by tracing out a specified set of respectively $m$ or $l$ qubits from the density matrix $\rho_{A}$ of the state $\ket{A}$.
If the state $\ket{A}$ is separable $\mathcal{S}_A = 0$.
Otherwise ($\mathcal{S}_A > 0$), the state is inseparable and for $\mathcal{S}_A = \min(l,m)$ it is maximally entangled. 
The eigenvalues of the density matrix \eqref{eq:reduced_density_matrix} can also be used to express the state \eqref{eq:state} using the Schmidt decomposition as
\begin{equation}\label{eq:Schmidt_decomposition}
\ket{A} = \sum_{j} \sqrt{\lambda_j} \ket{b_j} \otimes \ket{c_j},    
\end{equation}
where $\ket{b_j}$ and $\ket{c_j}$ are orthonormal bases of states $\ket{B}$ and $\ket{C}$.
Hence, the state \eqref{eq:Schmidt_decomposition} is separable if $\lambda_j=1$ for some $j$ and then $\ket{b_j}=\ket{B}$, $\ket{c_j}=\ket{C}$.
The entanglement between the states $\ket{B}$ and $\ket{C}$ is invariant under a unitary operation $U_{B} \otimes U_{C}$ acting on the individual states $\ket{B}$ and $\ket{C}$; the entanglement entropy \eqref{eq:entanglement_entropy} remains constant after such an operation.

States that can be written as tensor products of $s$ other pure states
($\ket{A} = \ket{B_1} \otimes \ket{B_2} \otimes \dots \otimes \ket{B_s}$)
are called \textit{$s$-separable} \cite{PhysRevA.61.042314}.
However, in this study, we focus on the bipartite separability condition \eqref{eq:separable_state} of pure $n$-qubit quantum registers $\ket{A_{n,k}}$, across at least one bipartition, and classify states according to the number and type of bipartitions across which they are separable, considering either equal or arbitrary, normalized PAs.

We explicitly exclude mixed states from the analysis, as they lack unique computational-basis support and define entanglement via convexity rather than support geometry, thereby requiring methods designed for mixed states, such as the Peres-Horodecki criterion or entanglement witnesses.
Furthermore, mixed states depend on the notion of classical probability

\section{Results}\label{sec:Results}

\begin{table*}[htbp]
\caption{\label{Table:Supports} 
The number of APS (blue) and PDS (three bipartitions - light green, one bipartition - dark green, no bipartitions - red) states of a quantum register containing $1 \le n \le 4$ qubits with arbitrary PAs with respect to the support size of the state (see text for details).
}
\begin{tabular}{|c|c|c|c|c|c|c|c|c|c|c|c|c|c|c|c|c|c|c|c|c|c|c|c|c|c|c|c|c|}
\hline
   & \multicolumn{27}{|c|}{Support Size of a Quantum State ($k$)} & \\
$n$&\cbb{1}&\cbb{2}&\cgr{2}&\cgb{2}&\crr{2}&\cgr{3}&\cgb{3}&\crr{3}&\cgr{4}&\cgb{4}&\crr{4}&\cgb{5}&\crr{5}&\cgb{6}&\crr{6}&\cgb{7}&\crr{7}&\cgb{8}&\crr{8}&\crr{9}&\crr{10}&\crr{11}&\crr{12}&\crr{13}&\crr{14}&\crr{15}&\crr{16}&$\Sigma$\\
\hline
1 &	2  & 1  &    &	  &	  &	   &     &     &    &     &      &     &      &     &      &    &       &   &       &       &      &      &      &     &     &    &   & 3 \\
2 &	4  & 4  &    &    & 2 &	   &     & 4   &    &     & 1    &     &      &     &      &    &       &   &       &       &      &      &      &     &     &    &   & 15 \\
3 &	8  & 12 &    & 12 & 4 &    & 24  & 32  &    & 6   & 64   &	   & 56   &     & 28   &    & 8     &   & 1     &       &      &      &      &     &     &    &   & 255 \\
4 &	16 & 32 & 48 & 32 &	8 & 96 & 256 & 208 & 24 & 512 &	1284 & 448 & 3920 & 224 & 7784 & 64 & 11376 & 8 & 12862 & 11440 & 8008 & 4368 & 1820 & 560 & 120 & 16 & 1 & 65535 \\
\hline
\end{tabular}
\end{table*}
\begin{table}[htbp]
\caption{\label{Table:Supports_PAs} The number of CMB$_c$ supports given by formula \eqref{eq:NofCMBc_supports} and the number of PDS (no bipartitions - red, one bipartition - dark green, three bipartitions - light  green) and APS (blue) states of a quantum register containing $1 \le n \le 4$ qubits (see text for details).}
\begin{tabular}{|c||c|c|c|c||c|c|c|c|}
\hline
    & \multicolumn{4}{|c||}{$\sum_{k=2}^{2^n} |\text{CMB}_c|(n,k)$} & \multicolumn{4}{|c|}{Arbitrary PAs} \\
$n$ & $c=0$ & $c=1$&$c=2$&$c=3$& \crr{PDS$_0$} & \cgb{PDS$_1$} & \cgr{PDS$_2$} & \cbb{APS} \\
\hline
1   & 1     & 0    & 0   & 0   &       &       &     & 3 \\
2   & 7     & 4    & 0   & 0   & 7     &       &     & 8 \\ 
3   & 193   & 42   & 12  & 0   & 193   & 42    &     & 20 \\
4   & 63775 & 1544 & 168 & 32  & 63775 & 1544  & 168 & 48 \\
\hline
\end{tabular}
\end{table}

\begin{lemma}\label{th:partitions}
The number of bipartitions a quantum register containing $n$ qubits can be separable across is $2^c-1$, where $0 \le c \le n-1$.
\end{lemma}
\begin{proof}
The total number of subsets of a set containing $n$ qubits is $2^n$, including the empty set and this set itself.
We have to exclude these two inseparable cases and also consider the symmetry of partitioning.
Hence, the upper bound on the number of unique bipartitions for $c=n-1$ is given by $(2^n-2)/2=2^{n-1}-1$.
An entanglement across one bipartition reduces the cardinality of the set to $n-1$ qubits that can be similarly partitioned across $(2^{n-1}-2)/2=2^{n-2}-1$ bipartitions.
Finally, an entanglement across all bipartitions can be thought of as a set containing only one element and thus inseparable.
\end{proof}

For example, the register of four qubits can be separable at most across
$2^{4-1}-1=7$
bipartitions as
$\{1|234\}$, $\{2|134\}$, $\{3|124\}$, $\{4|123\}$, $\{12|34\}$, $\{13|24\}$, and $\{14|23\}$,
where "$|$" denotes the bipartition.
However, an entanglement across one bipartition, say $\cancel{1|4} \coloneqq \text{V}$, reduces the cardinality of this set to three 
$\{2|3\text{V}\}$, $\{3|2\text{V}\}$, and $\{\text{V}|23\}$.

In the following definition, we can give meaning to the free parameter $c$ we introduced in Lemma~\ref{th:partitions}.

\begin{definition}\label{def:CBS}
We call a support a common-bit (CMB$_c$) support, where $0 \le c \le n-1$ is the largest integer such that there exist $c$ coordinates in which all $k \ge 2$ kets in the support have the same bit value.
\end{definition}
Geometrically, a CMB$_c$ support is an $(n-c)$-dimensional coordinate face of the Boolean hypercube, which under the Segre embedding \cite{cirici_characterization_2021} corresponds to a coordinate-fixed face, where $c$ coordinates are constant.

States that are separable across all bipartitions for all PAs are called \textit{fully separable} \cite{dur_separability_1999}.
We redefine them in the context of bipartitions.

\begin{definition}\label{def:APS}
We call the state $\ket{A}$ any-partition-separable (APS) if it is a one ket state or a state with CMB$_{n-1}$ support.
\end{definition}
For example, the following state with CMB$_{3}$ support
\begin{equation}
\begin{split}
&\ket{A_{4,2}} = \alpha_1 \ket{0_10_20_30_4} + \alpha_2 \ket{0_10_20_31_4} =\\
&= \ket{0_1} \otimes \left(\alpha_1 \ket{0_20_30_4} + \alpha_2 \ket{0_20_31_4} \right) = \ket{B_1}    \otimes \ket{C_{234}} =\\
&= \ket{0_2} \otimes \left(\alpha_1 \ket{0_10_30_4} + \alpha_2 \ket{0_10_31_4} \right) = \ket{B_2}    \otimes \ket{C_{134}} =\\
&= \ket{0_3} \otimes \left(\alpha_1 \ket{0_10_20_4} + \alpha_2 \ket{0_10_21_4} \right) = \ket{B_3}    \otimes \ket{C_{124}} =\\
&= \left(\alpha_1 \ket{0_4} + \alpha_2 \ket{1_4} \right) \otimes \ket{0_10_20_3}       = \ket{B_4}    \otimes \ket{C_{123}} =\\
&= \ket{0_10_2} \otimes \left(\alpha_1 \ket{0_30_4} + \alpha_2 \ket{0_31_4} \right)    = \ket{B_{12}} \otimes \ket{C_{34}}  =\\
&= \ket{0_10_3} \otimes \left(\alpha_1 \ket{0_20_4} + \alpha_2 \ket{0_21_4} \right)    = \ket{B_{13}} \otimes \ket{C_{24}}  =\\
&= \ket{0_20_3} \otimes \left(\alpha_1 \ket{0_10_4} + \alpha_2 \ket{0_11_4} \right)    = \ket{B_{23}} \otimes \ket{C_{14}} \\
\end{split}
\end{equation}
is an APS state as it is separable across all seven bipartitions for all PAs $\alpha_1$, $\alpha_2$.
As an APS state admits a superposition of at most one of its qubits, only states spanned over the vertices and 1-edges of $n$-cube are APS states.
If two or more qubits vary across the support, then some bipartition encounters differing values on both sides, forcing amplitude-dependent constraints.
As $n$-cube has $\binom{n}{m}2^{n-m}$ $m$-faces, a quantum register has
\begin{equation}
|\text{APS}(n)| = \binom{n}{0}2^{n-0} + \binom{n}{1}2^{n-1} = 2^{n-1}(n+2)
\end{equation}
APS states (vertices and blue edges in Fig.~\ref{Fig:support}).

The APS states introduce another definition.

\begin{definition}\label{def:PDS}
We call a state $\ket{A}$ a partition-dependent separable (PDS$_c$) state if it has a CMB$_c$ support, where $0 \le c \le n-2$.
\end{definition}

A state with a CMB$_c$ support is separable across $2^c-1$ bipartitions for all PAs.
Therefore, a state with a CMB$_0$ support is inseparable for all PAs.

\begin{theorem}\label{th:states_by_support_size}
The number of CMB$_c$ supports of an $n$-qubit quantum register having the support size $k \ge 2$ is given by
\begin{equation}\label{eq:NofCMBc_supports}
|\text{CMB}_c|(n,k) = \binom{n}{c}2^c \sum_{m=0}^{n-c-1} (-1)^m \binom{n-c}{m} 2^m \binom{2^{n-c-m}}{k}.
\end{equation}
\end{theorem}
\begin{proof}
Consider a set containing all $2^n$ distinct binary strings of length $n$.
$1 \le k \le 2^n$ bitstrings from this set can serve as a support of an $n$-qubit quantum register.
All one-element supports (vertices) are the APS states.
The remaining supports can have $0 \le c \le n-1$ common bit(s) in the same position(s).
There are $\binom{n}{c}$ ways to choose a subset of $c$ positions from $n$ positions and $2^c$ ways to assign bits to these positions.

The remaining $n-c-1$ positions must ensure that in none of them do all $k$ strings agree. 
The $k$ strings are identical in the $c$ positions, so choosing them is equivalent to selecting $k$ distinct vectors from Boolean space $\{0,1\}^{n-c}$ such that no coordinate in these vectors is constant for all $k$ vectors.
The total number of such $k$-subsets is $\binom{2^{n-c}}{k}$, but we need to subtract cases where at least one coordinate is constant, which can be done using the inclusion-exclusion principle.
Hence, we sum over $m$, the number of coordinates forced to be constant. 
There are $\binom{n-c}{m}$ ways to choose a subset of $m$ coordinates from $n-c$ coordinates and $2^m$ ways to assign bits to these coordinates.
The effective space size becomes $2^{n-c-m}$, so the count for those is $\binom{2^{n-c-m}}{k}$, with sign $(-1)^m$.
The inclusion–exclusion provides the sum in the formula \eqref{eq:NofCMBc_supports}, where $\binom{n}{k}=0$ for $0 \le n < k$ if the binomial coefficient is defined in terms of a falling factorial.
The formula \eqref{eq:NofCMBc_supports} counts the number of ways one can pick $1 \le k \le 2^{n-c}$ bitstrings from the set of all $2^n$ bitstrings of length $n$, so that each of these strings has only $0 \le c \le n-1$ bit(s) in common in the same position(s), which completes the proof.
\end{proof}

For example, for $n=3$ (cube), $k=2$ (all possible edges), and $c=1$ (only the face diagonals) the formula \eqref{eq:NofCMBc_supports} takes the form
\begin{equation}
\begin{split}
&|\text{CMB}_1|(3,2)
= \binom{3}{1}2^1 \sum_{m=0}^{3-1-1} (-1)^m \binom{3-1}{m} 2^m \binom{2^{3-1-m}}{2} = \\
&= 6 \left( 1 \cdot 1 \cdot 1 \cdot 6 - 1 \cdot 2 \cdot 2 \cdot 1  \right)= 6 \left( 6 - 4 \right) = 12,
\end{split}
\end{equation}
summing all six edges on each face of a cube (including face diagonals), subtracting four blue 1-edges having two common bits in the same position, and multiplying the result by six faces of the cube to count the states $\alpha_1\ket{000} + \alpha_4\ket{011}$, etc. (green edges in Fig.~\ref{Fig:support}(c)).

\bigskip
The distributions of $|\text{PDS}_c|(n,k)$ states are listed in the Table~\ref{Table:Supports} as functions of the support size $k$ and summed in the Table~\ref{Table:Supports_PAs} along with the values given by the formula \eqref{eq:NofCMBc_supports} for $1 \le n \le 4$ and $0 \le c \le n-1$. 
They were numerically cross-validated for $n \le 4$ (cf. Section~\ref{sec:Data_Accessibility}) by calculating the eigenvalues $\lambda_j$ of the reduced density matrices \eqref{eq:reduced_density_matrix} for each of $2^{2^n}-1$ quantum states corresponding to distinct supports for each of $2^{n-1}-1$ possible bipartitions, assuming equal or arbitrary PAs. 
The Schmidt decomposition \eqref{eq:Schmidt_decomposition} certifies that a given state is separable along a given bipartition if $\lambda_j \approx 1$ for some $j$.
For example, the PDS$_1$ state 
\begin{equation}
\ket{A_{3,4}} = \sqrt{\frac{1}{3}}\ket{100} + \sqrt{\frac{1}{6}}\ket{101} + \sqrt{\frac{1}{6}}\ket{110} + \sqrt{\frac{1}{3}}\ket{111}
\end{equation}
has the following eigenvalues of the reduced density matrix $\rho_{B}= \text{Tr}_C(\rho_{A_{3,4}})$ \eqref{eq:reduced_density_matrix}
\begin{equation}
\begin{split}\label{eq:numerical_ex}
\lambda_{\{1|23\}} &= \left\{0,      1 \right\},\\   
\lambda_{\{2|13\}} &= \lambda_{\{3|12\}} = \left\{\frac{3-2\sqrt{2}}{6}, \frac{3+2\sqrt{2}}{6}\right\},\\
\end{split}
\end{equation}
and thus it is separable only across the partition $\{1|23\}$, while for the remaining two partitions it has a fractional, weak entanglement entropy \eqref{eq:entanglement_entropy} $\mathcal{S}_A \approx 0.1873$.

\begin{lemma}\label{th:inseparable_for_all_PAs}
The number of states that are inseparable for all normalized PAs grows super-exponentially as a function of $n$ and for $n\ge2$ corresponds to the number of ways to choose a collection $\mathcal{C}$ of subsets $\mathcal{S}$ of $\mathcal{U} \in \{1,2,\dots,n\}$ such that $\cup\{\mathcal{S} \in \mathcal{C}\} = \mathcal{U}$ and $\cap\{\mathcal{S} \in \mathcal{C}\} = \emptyset$ (OEIS sequence \href{https://oeis.org/A131288}{A131288}).
\end{lemma}
\begin{proof}
The number of such supports is given by summing all $|\text{CMB}_0|(n,k)$ factors given by the formula \eqref{eq:NofCMBc_supports} over $k \ge 2$. Thus
\begin{equation}\label{eq:inseparable_for_all_PAs}
\begin{split}
&\sum_{k=2}^{2^n} |\text{CMB}_0|(n,k) = \\
&=\sum_{k=2}^{2^n} \binom{n}{0}2^0 \sum_{m=0}^{n-0-1} (-1)^m \binom{n-0}{m} 2^m \binom{2^{n-0-m}}{k} =\\
&=\sum_{m=0}^{n-1} (-1)^m \binom{n}{m} 2^m \sum_{k=2}^{2^n} \binom{2^{n-m}}{k} =\\
&=\sum_{m=0}^{n-1} (-1)^m \binom{n}{m} 2^m \left[ \sum_{k=0}^{2^n} \binom{2^{n-m}}{k} - \binom{2^{n-m}}{0} - \binom{2^{n-m}}{1} \right] =\\
&=\sum_{m=0}^{n-1} (-1)^m \binom{n}{m} 2^m \left( 2^{2^{n-m}}-1-2^{n-m} \right),\\
\end{split}
\end{equation}
simplifies to the formula of OEIS sequence \href{https://oeis.org/A131288}{A131288} $\{1_1,7_2,193_3,63775_4,\dots\}$ for $n\ge1$.
\end{proof}
For example, for $n=2$, the set $\mathcal{U} \in \{1,2\}$ has power set $\{\emptyset,\{1\},\{2\},\{1,2\}\}$ and there are seven ways to choose such a collection $\mathcal{C}$ of subsets $\mathcal{S}$ of $\mathcal{U}$ with full union and empty intersection corresponding to the states inseparable across all bipartitions for all PA values, as listed in Table~\ref{Table:A131288}. 
\begin{table}[htbp]
\caption{\label{Table:A131288} 
Collections $\mathcal{C}$ of subsets $\mathcal{S}$ of $\mathcal{U} \in \{1,2\}$ with full union and empty intersection and 2-qubit states inseparable for all PAs.}
\begin{tabular}{|c|c|ccccc|}
\hline
  & $\mathcal{C}$                       & State            & $\emptyset$        & $\{1\}$             & $\{2\}$             & $\{1,2\}$ \\
\hline
1 & $\{\emptyset,\{1\},\{2\},\{1,2\}\}$ & $\ket{A_{4,4}}=$ & $\alpha_1\ket{00}$ & $+\alpha_3\ket{10}$ & $+\alpha_2\ket{01}$ & $+\alpha_4\ket{11}$ \\
2 & $\{\emptyset,\{1\},\{2\}\}$         & $\ket{A_{4,3}}=$ & $\alpha_1\ket{00}$ & $+\alpha_3\ket{10}$ & $+\alpha_2\ket{01}$ &                    \\
3 & $\{\emptyset,\{2\},\{1,2\}\}$       & $\ket{A_{4,3}}=$ & $\alpha_1\ket{00}$ & $+\alpha_3\ket{10}$ &                     & $+\alpha_4\ket{11}$ \\
4 & $\{\emptyset,\{1\},\{1,2\}\}$       & $\ket{A_{4,3}}=$ & $\alpha_1\ket{00}$ &                     & $+\alpha_2\ket{01}$ & $+\alpha_4\ket{11}$ \\
5 & $\{\{1\},\{2\},\{1,2\}\}$           & $\ket{A_{4,3}}=$ &                    & $+\alpha_3\ket{10}$ & $+\alpha_2\ket{01}$ & $+\alpha_4\ket{11}$ \\
6 & $\{\emptyset, \{1,2\}\}$            & $\ket{A_{4,2}}=$ & $\alpha_1\ket{00}$ &                     &                     & $+\alpha_4\ket{11}$ \\
7 & $\{\{1\},\{2\}\}$                   & $\ket{A_{4,2}}=$ &                    & $+\alpha_3\ket{10}$ & $+\alpha_2\ket{01}$ &                    \\
\hline
\end{tabular}
\end{table}

\begin{lemma}
The maximum support size of a CMB$_c$ support is $k_{\text{max}}=2^{n-c}$.
\end{lemma}
\begin{proof}
Since $c$ out of $n$ bits are the same in all basis kets of the support, the remaining $n-c$ bits must be diversified in all the kets, and the maximum number of such sequences is $2^{n-c}$.
CMB$_c$ supports correspond to Hamming-weight-constrained subcubes.
\end{proof}
\begin{lemma}
The number of CMB$_c$ supports having the maximum support size $k_{\text{max}}$ is
\begin{equation}
|\text{CMB}_c|(n,k_{\text{max}}) = \binom{n}{c}2^c.
\end{equation}
\end{lemma}
\begin{proof}
This follows from substituting $k_{\text{max}}$ into the equation \eqref{eq:NofCMBc_supports}.
For example, there are $2n$ CMB$_1$ supports having such maximum support size, defined by ($n-1$)-dimensional facets of $n$-cube, as they have the largest support size for $n$ that can be partitioned across the same bipartition.
\end{proof}
For example, the maximum support size for the PDS$_1$ state of three qubits is $2^{3-1}=4$ and there are six states of the form
\begin{equation}\label{eq:max_support_size_example}
\begin{split}
&\ket{A_{3,4}}=\\
&=\alpha_1 \ket{0_10_20_3} + \alpha_2 \ket{0_10_21_3} + \alpha_3 \ket{0_11_20_3} + \alpha_4 \ket{0_11_21_3} = \\
&=\ket{0_1} \otimes \left( \alpha_1 \ket{0_20_3} + \alpha_2 \ket{0_21_3} + \alpha_3 \ket{1_20_3} + \alpha_4 \ket{1_21_3} \right) \\
\end{split}    
\end{equation}
in this case separable only across the $\{1|23\}$ bipartition for all normalized PAs.

Certain states are separable only for specific PAs.
Therefore, we introduce the last definition.

\begin{definition}\label{def:ADS}
We call the state $\ket{A}$ an amplitude-dependent separable (ADS$_{j,c}$) state, where $0 \le c < j \le n-1$, if its PAs can be arranged in a rank one $k_B \times k_C$ matrix.
\end{definition}
The PAs $\alpha$ matrix is an outer product of one column and one row matrices of PAs of the states $\ket{B}$ and $\ket{C}$ of the tensor product \eqref{eq:separable_state}
\begin{equation}\label{eq:alpha_matrix}
\begin{bmatrix}
\alpha_{1   1} & \alpha_{1   2} & \dots & \alpha_{1   k_C}\\
\alpha_{2   1} & \alpha_{2   2} & \dots & \alpha_{2   k_C}\\
\dots          & \dots          & \dots & \dots \\
\alpha_{k_B 1} & \alpha_{k_B 2} & \dots & \alpha_{k_B k_C}\\
\end{bmatrix}=
\begin{bmatrix}
\beta_1 \\ \beta_2 \\ \dots \\ \beta_{k_B}
\end{bmatrix}
\begin{bmatrix}
\gamma_1 & \gamma_2 & \dots & \gamma_{k_C}
\end{bmatrix}.
\end{equation}

States and supports defined in Definitions~\ref{def:CBS}-\ref{def:ADS} are listed in Fig.~\ref{Fig:taxonomy}.
Unlike APS and PDS$_c$ states, the separability of the ADS$_{j,c}$ is not a property of the support alone: the support merely permits separability, which is realized only when the PA matrix factors as a rank-1 outer product.
\begin{figure}[htbp]
\centering
\includegraphics[width=\columnwidth]{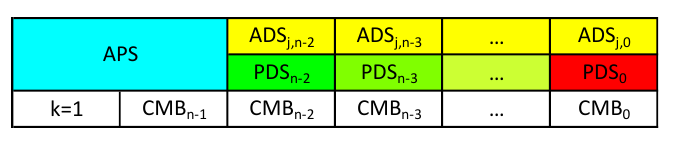}
\caption{\label{Fig:taxonomy}
Taxonomy of the quantum states and supports introduced in this study.
}
\end{figure}
\begin{figure*}[htbp]
\centering
\includegraphics[width=\textwidth]{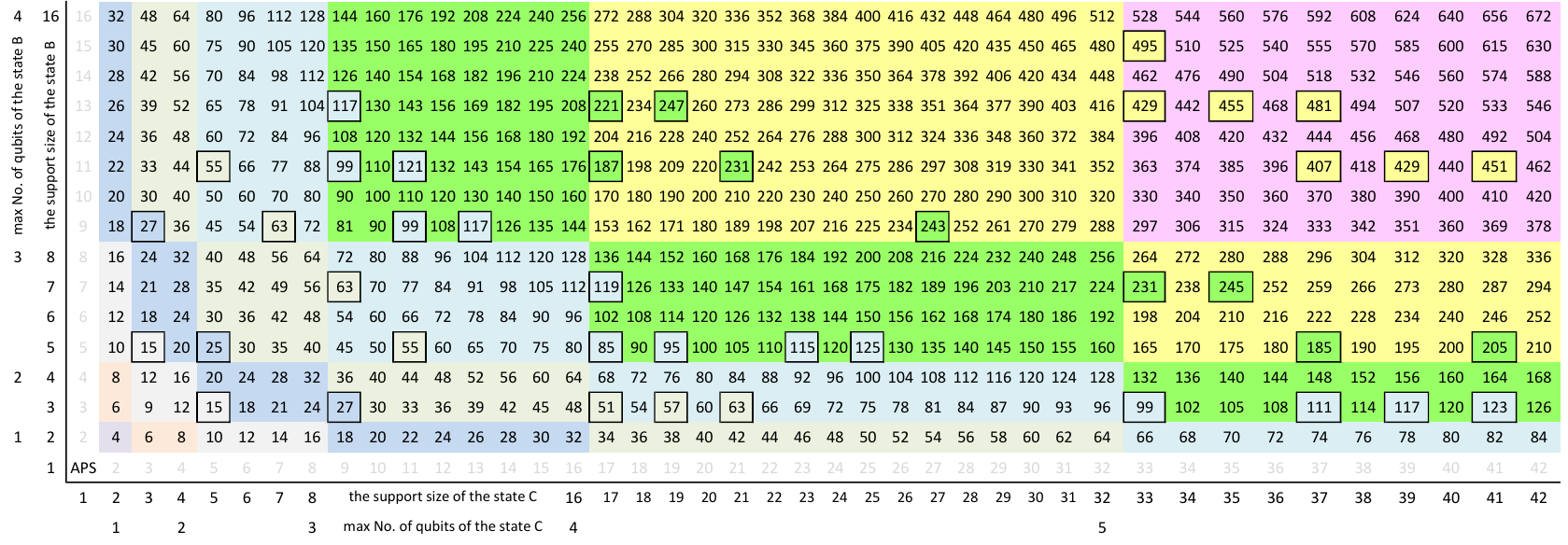}
\caption{\label{Fig:ADS}
Possible support sizes $k$ (products of the support sizes $k_B$ in the left column and $k_C$ in the bottom row) of ADS$_{c,j}$ states for $2 \le n \le 7$ qubits, indicated with different colors.
Forbidden sizes in boxes (OEIS sequence \href{https://oeis.org/A390536}{A390536}) are those where no bipartition can factor the support into $k_B × k_C$ satisfying the equation \eqref{eq:ADSineq} (see text for details).
}
\end{figure*}
\begin{lemma}\label{Th:ADS_possible_support_sizes}
An $n$-qubit state with support size $k$ can be an ADS state only if there exists a bipartition of the $n=l+m$ qubits ($l,m \ge 1$) and 
$k$ is a composite number satisfying
\begin{equation}\label{eq:ADSineq}
k = k_B \cdot k_C \le 2^{\lceil \log_2(k_B) \rceil + \lceil \log_2(k_C) \rceil},
\end{equation}
where $k_B, k_C \ge 2$ are the support sizes of these qubits.
\end{lemma}
\begin{proof}
The smallest composite number is four, so two qubits achieve the minimum support size of an ADS state ($k=4$).
Larger support sizes must be either even ($2\times\dots$) to provide separability across at least one bipartition or composite numbers satisfying
the inequalities $k_B \le 2^l$ and $k_C \le 2^m$ for $n=l+m$.
\end{proof}
\begin{lemma}\label{th:unconditionally_entangled}
A state having a support size $2^{n-1} < k \le 2^n$ that is prime or violates the inequality \eqref{eq:ADSineq} is unconditionally entangled; there is no bipartition allowing its separability even for amplitude‑dependent tuning.
\end{lemma}
For example, even though $3 \cdot 5 = 15 < 16 = 2^4$,
$\lceil \log_2(3) \rceil + \lceil \log_2(5) \rceil = 2+3=5 > 4 = \lceil \log_2(15) \rceil$: the support size of two qubits is at most four, while five kets are required for the second state in the product $15=3 \cdot 5$.
Such states correspond to a \textit{genuinely multipartite entangled} \cite{Palazuelos2022genuinemultipartite} or \textit{fully inseparable} \cite{dur_separability_1999} states, such as Bell state, Greenberger–Horne–Zeilinger (GHZ) state, or W-state.

\begin{lemma}\label{th:ADS}
An ADS$_{j,c}$ state is separable across $2^j - 2^c$ bipartitions (OEIS sequence \href{https://oeis.org/A023758}{A023758}).
\end{lemma}
\begin{proof}
We have to exclude PDS$_{c}$ states separable across $2^c-1$ bipartitions for all PAs (if any, i.e. for $c > 0$), from the larger set containing also the states separable across $2^j-1$ bipartitions both for all and for specific PAs to find 
$(2^j-1) - (2^c-1) = 2^j - 2^c$ PAs specific bipartitions.
In other words, $j-c$ indexes the \textit{depth} of amplitude-dependent separability.
\end{proof}
By Lemma~\ref{th:ADS}, PDS$_c$ and APS states are mutually exclusive.
A support of a state has an inherent $PDS_c$ classification, defining its separability for arbitrary PAs.
That same state may also become separable across additional partitions (that were previously entangled across) if its PAs are fine-tuned.

It is always possible to adjust the PAs of the $\alpha$ matrix of an ADS$_{j,c}$ state so that it has rank one. 
In particular, equal PAs make the $\alpha$ matrix constant, which factors as an outer product of two all-ones vectors $\beta$ and $\gamma$ scaled by $\alpha$, and any such outer-product matrix has rank one.

Possible support sizes $k$ of ADS states for $2 \le n \le 7$ are shown in Fig.~\ref{Fig:ADS} along with forbidden sizes violating the inequality \eqref{eq:ADSineq}.
Fig.~\ref{Fig:ADS} also shows support sizes providing separability across different numbers of bipartitions.
For example, a 4-qubit state with $k=12$ (gray zone) can be a $l=1$, $m=3$, $k=12=2 \cdot 6$ ADS$_{1,0}$ state separable across one bipartition or a $l=m=2$, $k=12=3 \cdot 4$ ADS$_{2,0}$ state separable across three bipartitions. 

For example, the PAs of the state 
\begin{equation}
\begin{split}
&\ket{A_{4,10}} = \\
 &\alpha_{0000} \ket{0000} + \alpha_{0001} \ket{0001} + \alpha_{0010} \ket{0010} + \alpha_{0011} \ket{0011} +\\
+&\alpha_{0100} \ket{0100} + \alpha_{0101} \ket{0101} + \alpha_{0110} \ket{0110} + \alpha_{0111} \ket{0111} +\\
+&\alpha_{1010} \ket{1010} + \alpha_{1011} \ket{1011} \\
\end{split}    
\end{equation}
can be written as an outer product
\begin{equation}\label{eq:alpha_matrix65328}
\begin{bmatrix}
\alpha_{0000} & \alpha_{0001} \\
\alpha_{0010} & \alpha_{0011} \\
\alpha_{0100} & \alpha_{0101} \\
\alpha_{0110} & \alpha_{0111} \\
\alpha_{1010} & \alpha_{1011} \\
\end{bmatrix}=
\begin{bmatrix}
\beta_{000} \\ \beta_{001} \\ \beta_{010} \\ \beta_{011} \\ \beta_{101} 
\end{bmatrix} 
\begin{bmatrix}
\gamma_0 & \gamma_1
\end{bmatrix}
\end{equation}
as it is is an ADS$_{1,0}$ state separable across the bipartition $\{123|4\}$ as
\begin{equation}
\begin{split}
\ket{A_{4,10}} &=
\left( \beta_{000} \ket{000} + \beta_{001} \ket{001} + \beta_{010} \ket{010} + \right. \\
&+\left. \beta_{011} \ket{011} + \beta_{101} \ket{101} \right) \otimes \left( \gamma_0 \ket{0} + \gamma_1 \ket{1} \right)
\end{split}    
\end{equation}
if and only
\begin{equation}
  \frac{\alpha_{0000}}{\alpha_{0001}}
= \frac{\alpha_{0010}}{\alpha_{0011}}
= \frac{\alpha_{0100}}{\alpha_{0101}}
= \frac{\alpha_{0110}}{\alpha_{0111}}
= \frac{\alpha_{1010}}{\alpha_{1011}}
= \frac{\gamma_0}{\gamma_1},
\end{equation}
that is, when the columns of the $\alpha$ matrix \eqref{eq:alpha_matrix65328} are linearly dependent or equivalently if its rank is 1.
Similarly, the PAs of the state
\begin{equation}
\begin{split}
\ket{A_{4,9}} = &\alpha_{0101} \ket{0101} + \alpha_{0110} \ket{0110} + \alpha_{0111} \ket{0111} +\\
          &\alpha_{1001} \ket{1001} + \alpha_{1010} \ket{1010} + \alpha_{1011} \ket{1011} +\\
          &\alpha_{1101} \ket{1101} + \alpha_{1110} \ket{1110} + \alpha_{1111} \ket{1111}, \\          
\end{split}
\end{equation}
can be written as an outer product
\begin{equation}\label{eq:alpha_matrix1911}
\begin{bmatrix}
\alpha_{0101} & \alpha_{0110} & \alpha_{0111} \\
\alpha_{1001} & \alpha_{1010} & \alpha_{1011} \\
\alpha_{1101} & \alpha_{1110} & \alpha_{1111} \\          
\end{bmatrix}=
\begin{bmatrix}
\beta_{01} \\ \beta_{10} \\ \beta_{11}
\end{bmatrix} 
\begin{bmatrix}
\gamma_{01} & \gamma_{10} & \gamma_{11}
\end{bmatrix}
\end{equation}
and it is also the ADS$_{1,0}$ state separable only across the bipartition $\{12|34\}$ as
\begin{equation}
\begin{split}
\ket{A_{4,9}} &=\left(\beta_{01} \ket{01} + \beta_{10}\ket{10} + \beta_{11}\ket{11} \right) \\
          &\otimes \left( \gamma_{01}\ket{01} + \gamma_{10}\ket{10} + \gamma_{11}\ket{11} \right)\\          
\end{split}
\end{equation}
again if all the rows and columns of the $\alpha$ matrix \eqref{eq:alpha_matrix1911} are linearly independent (its rank is 1).
The state (we omit kets here for clarity)
\begin{equation}
\begin{split}
\ket{A_{4,4}} 
&= \alpha_{0000} + \alpha_{0001} + \alpha_{0110} + \alpha_{0111} =\\
&=\left( \beta_{000}  + \beta_{011} \right) \otimes \left( \gamma_{0} + \gamma_{1} \right)=\\        
&=\left( \beta_{0_20_3} + \beta_{1_21_3} \right) \otimes \left( \gamma_{0_10_4} + \gamma_{0_11_4} \right)\\  \end{split}
\end{equation}
is a PDS$_1$ support separable across the bipartition $\{1|234\}$ and 
it is also a ADS$_{2,1}$ state separable across two bipartitions $\{123|4\}$ and $\{23|14\}$, if the $\alpha$ matrix
\begin{equation}\label{eq:alpha_matrix49920}
\begin{bmatrix}
\alpha_{0000} & \alpha_{0001} \\
\alpha_{0110} & \alpha_{0111} \\
\end{bmatrix}=
\begin{bmatrix}
\beta_{0_10_20_3} \\ \beta_{0_11_21_3}
\end{bmatrix} 
\begin{bmatrix}
\gamma_{0_4} & \gamma_{1_4}
\end{bmatrix}=
\begin{bmatrix}
\beta_{0_20_3} \\ \beta_{1_21_3}
\end{bmatrix} 
\begin{bmatrix}
\gamma_{0_10_4} & \gamma_{0_11_4}
\end{bmatrix}
\end{equation}
has a rank of one, that is if
\begin{equation}
\alpha_{0000}\alpha_{0111} = \alpha_{0110}\alpha_{0001}.
\end{equation}

Any CMB$_c$ support having the maximum support size $2^{n-c}$ has basis kets with bits differing only in the same $u \coloneqq n-c \ge 2$ positions (i.e., spanning $2^u$ vertices of $u$-cube).
In this case, we can define $\alpha_{\hat{x}}$ as the PA corresponding to the ket being the bitwise complement of the relevant $u$ positions of the ket associated with a PA $\alpha_{x}$ and the separability condition for the PAs $\alpha_{x}$ and $\alpha_{\hat{x}}$ is
\begin{equation}\label{eq:separability_condition}
\alpha_{x} \alpha_{\hat{x}} = \text{const} \ne 0, \quad \forall x \in \{0,1\}^u.
\end{equation}
In other words, the products of the PAs associated with the vertices defining all the $2^{u-1}$ longest diagonals of the $u$-cube must be equal to each other.
In particular
\begin{equation}\label{eq:separability_condition_part}
\begin{array}{ll}
\alpha_{0}                 = \alpha_{1}                                                                             &\text{~if~} u=1, \\
\alpha_{00}\alpha_{11}     = \alpha_{01}\alpha_{10}                                                                 &\text{~if~} u=2, \\ 
\alpha_{000}\alpha_{111}   = \alpha_{001}\alpha_{110}   = \alpha_{010}\alpha_{101}   = \alpha_{100}\alpha_{011}     &\text{~if~} u=3, \\ 
\alpha_{0000}\alpha_{1111} = \alpha_{0001}\alpha_{1110} = \alpha_{0010}\alpha_{1101} = \alpha_{0011}\alpha_{1100} = &\\ 
\alpha_{0100}\alpha_{1011} = \alpha_{0101}\alpha_{1010} = \alpha_{0110}\alpha_{1001} = \alpha_{0111}\alpha_{1000}   &\text{~if~} u=4. \\ 
\dots
\end{array}
\end{equation}

For example, the state
\begin{equation}
\begin{split}
\ket{A_{4,4}} = &\alpha_{1100} \ket{1_1 1_2 0_3 0_4} + \alpha_{1101} \ket{1_1 1_2 0_3 1_4} +\\
          &\alpha_{1110} \ket{1_1 1_2 1_3 0_4} + \alpha_{1111} \ket{1_1 1_2 1_3 1_4} = \\
=& \ket{B_1} \otimes \ket{C_{234}} = \ket{B_2} \otimes \ket{C_{134}} = \ket{B_{12}} \otimes \ket{C_{34}}\\          
\end{split}
\end{equation}
has two common bits in the same positions 1 and 2, so it is a CMB$_2$ support separable across $2^2-1=3$ bipartitions $\{1|234\}$, $\{2|134\}$, or $\{12|34\}$ for all PAs. 
But if
\begin{equation}
\frac{\alpha_{1101}}{\alpha_{1100}} =
\frac{\alpha_{1111}}{\alpha_{1101}} =
\frac{\gamma_{10}} {\gamma_{11}} = 
\frac{\gamma_{111}}{\gamma_{110}}.
\end{equation}
it is also an ADS$_{3,2}$ state separable across four bipartitions $\{3|124\}$, $\{4|123\}$, $\{13|24\}$, and $\{14|23\}$ separable as
\begin{equation}
\begin{split}
&\ket{A} =\\
&=\left(\beta_{0} \ket{0_3} + \beta_{1}\ket{1_3} \right) \otimes \left( \gamma_{110}\ket{1_1 1_2 0_4} + \gamma_{111}\ket{1_1 1_2 1_4} \right) =\\
&=\left(\beta_{0} \ket{0_4} + \beta_{1}\ket{1_4} \right) \otimes \left( \gamma_{110}\ket{1_1 1_2 0_3} + \gamma_{111}\ket{1_1 1_2 1_3} \right) =\\
&=\left(\beta_{10} \ket{1_1 0_3} + \beta_{11}\ket{1_1 1_3} \right) \otimes \left( \gamma_{10}\ket{1_2 0_4} + \gamma_{11}\ket{1_2 1_4} \right) =\\
&=\left(\beta_{10} \ket{1_1 0_4} + \beta_{11}\ket{1_1 1_4} \right) \otimes \left( \gamma_{10}\ket{1_2 0_3} + \gamma_{11}\ket{1_2 1_3} \right).\\
\end{split}
\end{equation}
For the same reasons, the state \eqref{eq:numerical_ex} is separable across all three bipartitions after swapping PAS $\alpha_{110}$ with $\alpha_{111}$.

For example, the following state supported on four basis kets with two qubits spanned over four vertices of the 2-cube
\begin{equation}
\begin{split}
\ket{A_{4,4}} 
&=\alpha_{0_10_2} \ket{0_1 1 0_2 1} + \alpha_{0_11_2} \ket{0_1 1 1_2 1} + \\
&+\alpha_{1_10_2} \ket{1_1 1 0_2 1} + \alpha_{1_11_2} \ket{1_1 1 1_2 1} = \\
&=\left(\beta_{0_1}\ket{0_1 1}         + \beta_{1_1}     \ket{1_1 1}         \right) \otimes 
  \left(\gamma_{0_2}          \ket{0_2 1} + \gamma_{1_2}               \ket{1_2 1} \right)
\end{split}
\end{equation}
has a CMB$_2$ support separable across three bipartitions $\{2|134\}$, $\{13|24\}$, and $\{4|123\}$.
However, if 
$\alpha_{0_10_2} \alpha_{1_11_2} = \alpha_{0_11_2} \alpha_{1_10_2}$,
it is also an ADS$_{3,2}$ state separable across four bipartitions $\{1|234\}$, $\{12|34\}$, $\{3|124\}$, and $\{14|23\}$.

A 3-qubit register (we omit kets here for clarity)
\begin{equation}
\begin{split}
\ket{A_{4,8}} 
&= \alpha_{000}\dots + \alpha_{001}\dots + \alpha_{010}\dots + \alpha_{011}\dots \\
&+ \alpha_{100}\dots + \alpha_{101}\dots + \alpha_{110}\dots + \alpha_{111}\dots =\\
&\left(\beta_{0} + \beta_{1} \right) \otimes \left( \gamma_{00} + \gamma_{01} + \gamma_{10} + \gamma_{11} \right) \\
\end{split}
\end{equation}
can be an ADS$_{1,0}$ state separable across one bipartition $(\{1|23\})$
if the $\alpha$ matrix 
\begin{equation}\label{eq:alpha_matrix49920_1|23}
\begin{bmatrix}
\alpha_{\textcolor{myblue}{0}00} & \alpha_{\textcolor{myblue}{0}01} & \alpha_{\textcolor{myblue}{0}10} & \alpha_{\textcolor{myblue}{0}11} \\
\alpha_{\textcolor{myblue}{1}00} & \alpha_{\textcolor{myblue}{1}01} & \alpha_{\textcolor{myblue}{1}10} & \alpha_{\textcolor{myblue}{1}11} \\
\end{bmatrix}=
\begin{bmatrix}
\beta_{0} \\ \beta_{1}
\end{bmatrix} 
\begin{bmatrix}
\gamma_{00} & \gamma_{01} & \gamma_{10} & \gamma_{11}
\end{bmatrix}
\end{equation}
has rank one,
which is equivalent to 
\begin{equation}
\begin{split}
& \alpha_{000}\alpha_{101} = \alpha_{001}\alpha_{100} \quad \land \quad 
  \alpha_{000}\alpha_{110} = \alpha_{010}\alpha_{100} \quad \land \\ 
& \alpha_{000}\alpha_{111} = \alpha_{011}\alpha_{100} \quad \land \quad 
  \alpha_{001}\alpha_{110} = \alpha_{010}\alpha_{101} \quad \land \\ 
& \alpha_{001}\alpha_{111} = \alpha_{011}\alpha_{101} \quad \land \quad 
  \alpha_{010}\alpha_{111} = \alpha_{011}\alpha_{110}, \\ 
\end{split}    
\end{equation}
across one bipartition $(\{2|13\})$
if the $\alpha$ matrix 
\begin{equation}\label{eq:alpha_matrix49920_2|13}
\begin{bmatrix}
\alpha_{0\textcolor{myblue}{0}0} & \alpha_{0\textcolor{myblue}{0}1} & \alpha_{1\textcolor{myblue}{0}0} & \alpha_{1\textcolor{myblue}{0}1} \\
\alpha_{0\textcolor{myblue}{1}0} & \alpha_{0\textcolor{myblue}{1}1} & \alpha_{1\textcolor{myblue}{1}0} & \alpha_{1\textcolor{myblue}{1}1} \\
\end{bmatrix}=
\begin{bmatrix}
\beta_{0} \\ \beta_{1}
\end{bmatrix} 
\begin{bmatrix}
\gamma_{00} & \gamma_{01} & \gamma_{10} & \gamma_{11}
\end{bmatrix}
\end{equation}
has rank one, and
across one bipartition $(\{3|12\})$
if the $\alpha$ matrix 
\begin{equation}\label{eq:alpha_matrix49920_3|12}
\begin{bmatrix}
\alpha_{00\textcolor{myblue}{0}} & \alpha_{01\textcolor{myblue}{0}} & \alpha_{10\textcolor{myblue}{0}} & \alpha_{11\textcolor{myblue}{0}} \\
\alpha_{00\textcolor{myblue}{1}} & \alpha_{01\textcolor{myblue}{1}} & \alpha_{10\textcolor{myblue}{1}} & \alpha_{11\textcolor{myblue}{1}} \\
\end{bmatrix}=
\begin{bmatrix}
\beta_{0} \\ \beta_{1}
\end{bmatrix} 
\begin{bmatrix}
\gamma_{00} & \gamma_{01} & \gamma_{10} & \gamma_{11}
\end{bmatrix}
\end{equation}
has rank one.
It can also be the ADS$_{2,0}$ state separable across all $2^{2}-1=3$ bipartitions if 
$\alpha_{000}\alpha_{111} = \alpha_{001}\alpha_{110} = \alpha_{010}\alpha_{101} = \alpha_{011}\alpha_{100}$, i.e., if the products for all four main 3-cube diagonals are equal.
The 4-qubit register 
\begin{equation}
\begin{split}
&\ket{A} 
= \alpha_{0000}\dots + \alpha_{0001}\dots + \alpha_{0010}\dots + \alpha_{0011}\dots \\
&+ \alpha_{0100}\dots + \alpha_{0101}\dots + \alpha_{0110}\dots + \alpha_{0111}\dots \\
&+ \alpha_{1000}\dots + \alpha_{1001}\dots + \alpha_{1010}\dots + \alpha_{1011}\dots \\
&+ \alpha_{1100}\dots + \alpha_{1101}\dots + \alpha_{1110}\dots + \alpha_{1111}\dots = \\
&= \left( \beta_{0} + \beta_{1} \right) \otimes\\
&  \left(\gamma_{000} + \gamma_{001} + \gamma_{010} + \gamma_{011} + \gamma_{100} + \gamma_{101} + \gamma_{110} + \gamma_{111} \right) \\
&= \left( \beta_{00} + \beta_{01} + \beta_{10} + \beta_{11} \right) \otimes \left(\gamma_{00} + \gamma_{01} + \gamma_{10} + \gamma_{11} \right)
\end{split}
\end{equation}
can be an ADS$_{1,0}$ state separable across one of four possible bipartitions ($\{1|234\}$ is shown below) if the $\alpha$ matrix is factorable as
\begin{equation}\label{eq:alpha_matrix_1|234}
\begin{split}
&
\begin{bmatrix}
\alpha_{\textcolor{myblue}{0}000} & \alpha_{\textcolor{myblue}{0}001} & \alpha_{\textcolor{myblue}{0}010} & \alpha_{\textcolor{myblue}{0}011} &
\alpha_{\textcolor{myblue}{0}100} & \alpha_{\textcolor{myblue}{0}101} & \alpha_{\textcolor{myblue}{0}110} & \alpha_{\textcolor{myblue}{0}111} \\
\alpha_{\textcolor{myblue}{1}000} & \alpha_{\textcolor{myblue}{1}001} & \alpha_{\textcolor{myblue}{1}010} & \alpha_{\textcolor{myblue}{1}011} &
\alpha_{\textcolor{myblue}{1}100} & \alpha_{\textcolor{myblue}{1}101} & \alpha_{\textcolor{myblue}{1}110} & \alpha_{\textcolor{myblue}{1}111}\\
\end{bmatrix}=\\&
\begin{bmatrix}
\beta_{0} \\ \beta_{1}
\end{bmatrix} 
\begin{bmatrix}
\gamma_{000} & \gamma_{001} & \gamma_{010}  & \gamma_{011} & \gamma_{100} & \gamma_{101} & \gamma_{110} & \gamma_{111}
\end{bmatrix}
\end{split}
\end{equation}
an ADS$_{1,0}$ state separable across one of three bipartitions ($\{12|34\}$ is shown below) if the $\alpha$ matrix is factorable as
\begin{equation}\label{eq:alpha_matrix_12|34}
\begin{bmatrix}
\alpha_{\textcolor{myblue}{00}00} & \alpha_{\textcolor{myblue}{00}01} & \alpha_{\textcolor{myblue}{00}10} & \alpha_{\textcolor{myblue}{00}11} \\
\alpha_{\textcolor{myblue}{01}00} & \alpha_{\textcolor{myblue}{01}01} & \alpha_{\textcolor{myblue}{01}10} & \alpha_{\textcolor{myblue}{01}11} \\
\alpha_{\textcolor{myblue}{10}00} & \alpha_{\textcolor{myblue}{10}01} & \alpha_{\textcolor{myblue}{10}10} & \alpha_{\textcolor{myblue}{10}11} \\
\alpha_{\textcolor{myblue}{11}00} & \alpha_{\textcolor{myblue}{11}01} & \alpha_{\textcolor{myblue}{11}10} & \alpha_{\textcolor{myblue}{11}11} \\
\end{bmatrix}=
\begin{bmatrix}
\beta_{00} \\ \beta_{01} \\ \beta_{10} \\ \beta_{11}
\end{bmatrix} 
\begin{bmatrix}
\gamma_{00} & \gamma_{01} & \gamma_{10}  & \gamma_{11}
\end{bmatrix}
\end{equation}
and an ADS$_{3,0}$ state separable across all seven bipartitions if all PAs are equal to each other.

A $2$-cube (square) is the smallest $n$-cube that provides the support size for the ADS$_{1,1}$ state.
Even though one qubit ($1$-cube, segment) is always separable, the separability condition \eqref{eq:separability_condition} can be extended to this case ($j=1$), where it implies the equality of two PAs \eqref{eq:separability_condition_part}.
This is a specific case of the condition of equal superposition of a qubit; the ADS$_{1,0}$ state
\begin{equation}\label{eq:qubit}
\ket{A_{1,2}} = \frac{1}{\sqrt{2}}\left( \ket{0} + \ket{1} \right),
\end{equation}
with a vanishing relative phase between the basis kets:
on the complex plane $\alpha_{0} = \alpha_{1}$ represents the same point on the circle of radius $1/\sqrt{2}$, while $|\alpha_{0}|^2=|\alpha_{1}|^2=1/2$ represents all points of this circle.

\section{Applications}\label{sec:Applications}

In general, determining across which bipartition (if any) a state is separable is computationally exponential.
The computational complexity for finding all eigenvalues of a Hermitian matrix of size $M$ is $O(M^3)$.
In our case $M=2^{\lfloor n/2 \rfloor} \approx 2^{n/2}$ for the maximum number of qubits of the state $\ket{A}$ or $\ket{B}$ in the tensor product \eqref{eq:separable_state}, yielding $O(8^{n/2})$. Calculating the reduced density matrix also requires $O(2^{n+\lfloor n/2 \rfloor}) \approx O(8^{n/2})$ operations.
This yields an exponential complexity of $(2^{n-1}-1)O(8^{n/2}) = O(\sqrt{32}^n) = O(2^{2.5n})$, for all $2^{n-1}-1$ bipartitions to check.
On the other hand, determining the distribution of bits in the same positions across $k$ basis kets of an $n$-qubit state, required to classify a support of a state as PDS$_c$ has computational complexity of only $O(nk)$.
The speedup ($2^{2.5n}/nk$) thus grows exponentially with $n$ for all $k \le 2^n$ as a complexity ratio.
It is particularly useful, however, in the sparse-support regime $k \ll 2^{n/2}$, where the support geometry is informative, as is typical for many structured, post-selected, or shallow-circuit states.

The support structure alone provides immediate separability bounds, allowing one to bypass expensive density-matrix computations.
In particular, $c \ne 0$ is a direct parameter of an \textit{entanglement confinement}, as the state always has the form of 
\begin{equation}\label{eq:PDSc_form}
\ket{A_{n,k}} = \ket{B_{n-c,k_B}} \otimes \ket{C^\text{classical}_{c,k_C}}
\end{equation}
containing $c$ classical (i.e., non-superposed) qubits up to relabeling.

Furthermore, states having support sizes of Lemma~\ref{th:unconditionally_entangled} cannot be realized by any separable bipartition even if PAs are equal.
As a consequence, a uniform superposition over any 15 computational basis states in 4 qubits, for example, must be entangled across every possible bipartition — it is a genuinely multipartite entangled state \cite{Palazuelos2022genuinemultipartite}.
This provides an extremely simple way to construct states with guaranteed full multipartite entanglement by choosing any support of a forbidden size (e.g., with equal PAs).

The support taxonomy can be used as a preprocessing tool:
\begin{enumerate}
\item if the state is of APS or PDS$c$ type, its separability is known without Schmidt analysis,
\item if the state support size is forbidden (Lemma~\ref{th:unconditionally_entangled}), its unconditional entanglement is known without Schmidt analysis, and only
\item if the state is of ADS type, Schmidt analysis is required.  
\end{enumerate}
Hence, the support taxonomy does not replace Schmidt ranks; rather, it restricts their necessity to ADS states.

Further exemplary applications of the introduced taxonomy are provided below.

\subsection{Tracking entanglement spreading in quantum circuits}

Consider the standard textbook circuit that creates an $n$-qubit GHZ state which starts with the $\ket{0_10_2\dots 0_n}$ APS state, applies the Hadamard gate to the first qubit to get the CMB$_{n-1}$ support (still APS state), and then applies a chain of CNOT gates to arrive at the final CMB$_0$ state, as shown below
\begin{equation}\label{eq:Entanglement_Spreading_Example}
\begin{split}
\ket{A_{n,1}} &= \ket{0_10_2\dots 0_{n}} \myH \\
\ket{A_{n,2}} &= \frac{1}{\sqrt{2}}\left(\ket{0_1} + \ket{1_1}\right) \otimes \ket{0_20_3\dots 0_{n}} \myC \\
\ket{A_{n,2}} &= \frac{1}{\sqrt{2}}\left(\ket{0_10_2} + \ket{1_11_2}\right) \otimes \ket{0_31_4\dots 0_{n}} \myC \\
\ket{A_{n,2}} &= \frac{1}{\sqrt{2}}\left(\ket{0_10_20_3} + \ket{1_11_21_3}\right) \otimes \ket{0_40_5\dots 0_{n}} \myC \\
&\dots \\
\ket{A_{n,2}} &= \frac{1}{\sqrt{2}}\left(\ket{0_10_2\dots 0_n} + \ket{1_11_2\dots 1_n}\right). \\
\end{split}    
\end{equation}
Each operation reduces the number of common bits $c$ by exactly 1 and delocalizes the entanglement to one additional qubit.
The taxonomy proposed in this study, therefore, gives an exact, single-integer metric that tracks how entanglement spreads through the CNOT chain - something that would otherwise require expensive Schmidt decompositions across all $2^{n-1}-1$ bipartitions.
This is particularly useful for analyzing fault-tolerance thresholds in GHZ preparation or distribution: if an error occurs when $c$ is still large, the entanglement is still localized to a small block, so the error may be easier to correct locally.
This is illustrated in the equation \eqref{eq:Entanglement_Spreading_Example}, where the fourth qubit erroneously flipped after the first CNOT gate, is corrected before the second one. 

\subsection{Classical simulation speedup for localized-entanglement states}

Because high-$c$ states are in the form \eqref{eq:PDSc_form}, any quantum circuit that produces a state with large $c$ can be classically simulated by discarding the $c$ \textit{classical} qubits and only simulating the remaining $n-c$ qubits.
Even for arbitrary PAs, the fixed qubits contribute only a global factor and never entangle with anything.
Consider, for example, a circuit that produces an $10$-qubit state with $c=7$.
Then the state is a fixed classical $7$-bit string on $7$ qubits tensored with an entangled $3$-qubit state, and the unitary evolution can be simulated on just $3$ qubits instead of $10$, providing an exponential savings.
This feature applies directly to:
\begin{itemize}
\item sparse-state or low-weight circuits,
\item post-selected computations,
\item certain variational ansätze that accidentally stay in high-$c$ supports,
\item debugging shallow circuits where support hasn't fully spread, etc.
\end{itemize}

The $O(nk)$ check for $c$ is vastly cheaper than computing entanglement entropies across all bipartitions, making it a practical pre-filter for tensor-network or exact simulators: "if $c \ge$ some threshold, reduce to $n-c$ qubits".

In summary, the taxonomy turns the hypercube support into a powerful diagnostic tool that is both theoretically clean and practically cheap, especially for tracking entanglement dynamics in circuits, enabling aggressive classical simulation when entanglement stays localized, and constructing guaranteed fully inseparable states.
The taxonomy enables both the detection of entanglement properties in existing states and the design of states with desired entanglement characteristics

\section{Conclusions}\label{sec:Conclusions}

Classifying quantum states using $n$-cube separability structures provides a fast ($O(nk)$ vs. $O(2^{2.5n})$) structural way to assess entanglement without full-state tomography or diagonalization.
Other potential applications in quantum computing include quantum machine learning and data encoding, quantum circuit optimization, quantum communication and network protocols, and quantum-to-classical boundary studies.

\section{Data Accessibility}\label{sec:Data_Accessibility}
The public repository for the code written in the MATLAB computational environment to cross-validated the distributions listed in Tables~\ref{Table:Supports} and \ref{Table:Supports_PAs}, and the XLS file containing them is given under the link 
\url{https://github.com/szluk/n-qubes} (accessed on 19 September 2025).

\section*{Acknowledgments}
I thank my partners Wawrzyniec Bieniawski and Piotr Masierak for their numerous clarifications, formal corrections, and improvements.




\bibliographystyle{vancouver}
\bibliography{apssamp_all}

\end{document}